\documentclass{article}
\usepackage[utf8]{inputenc}
\usepackage{indentfirst}
\usepackage{graphicx,tikz, tikz-cd}
\usetikzlibrary{decorations.pathreplacing}
\usetikzlibrary{arrows}
\usepackage{fancyhdr}
\usepackage{hyperref}

\hypersetup{pdftitle={Entanglement},pdfauthor={Lorenzo Panebianco}}
\graphicspath{ {images/} }
\usepackage[T1]{fontenc}
\usepackage{mathrsfs}
\usepackage[english]{babel}
\usepackage{float,floatflt, epsfig}
\usepackage{subfig}
\usepackage{color}
\usepackage{caption}
\usepackage{upgreek}
\usepackage{lipsum} 
\usepackage{latexsym}
\usepackage{mathtools}

\usepackage{amssymb,amsthm,amsmath}
\usepackage{bbm}
\usepackage{amscd,wasysym}
\usepackage{braket}
\usepackage[toc,page]{appendix}
\usepackage{enumerate} 
\usepackage{booktabs}
\usepackage{mathtools}
\usepackage[all]{xy}
\usepackage[margin=1.1in]{geometry}

\newcommand{\vertiii}[1]{{\left\vert\kern-0.25ex\left\vert\kern-0.25ex\left\vert #1 
		\right\vert\kern-0.25ex\right\vert\kern-0.25ex\right\vert}}

\newcommand{\R}{\mathbb{R}}

\newcommand{\hil}{\mathcal H}

\newcommand{\hi}{\mathcal{H}}
\newcommand{\al}{\mathcal{A}}

\newcommand{\M}{\mathcal{M}}

\theoremstyle{plain}
\newtheorem{thm}{Theorem}

\newtheorem*{thm*}{Theorem}
\newtheorem{cor}[thm]{Corollary}
\newtheorem*{cor*}{Corollary}

\newtheorem{lem}[thm]{Lemma} 
\newtheorem*{lem2a*}{Lemma 2A} 
\newtheorem*{lem2b*}{Lemma 2B} 
\newtheorem*{lem2c*}{Lemma 2C} 
\newtheorem{prop}[thm]{Proposition}

\newtheorem*{claim*}{Claim} 

\theoremstyle{definition}

\newtheorem{defn}[thm]{Definition}

\newtheorem*{defn*}{Definition}

\theoremstyle{remark}
\newtheorem{remark}[thm]{Remark}

\title{Entanglement Entropy in CFT and Modular Nuclearity}

\author{{Lorenzo Panebianco}${}^{1}$\footnote{{\tt lorenzopnb@proton.me} , ${}^\triangle${\tt benny.wegener@gmx.de} }, {Benedikt Wegener}${}^{2,3,\triangle}$.
}
\date{\small{
		${}^{1}$Dipartimento di Matematica,  Università di Roma “La Sapienza” \\
		Piazzale Aldo Moro 5, 00185–Roma, Italy \\
		${}^{2}$ Marie Sklodowska-Curie fellow of the Istituto Nazionale di Alta Matematica\\
		${}^{3}$ Dipartimento di Matematica, Universit\`a di Roma “Tor Vergata” \\
		Via della Ricerca Scientifica 1, I-00133 Roma, Italy}\\
}

\begin{document}

	\maketitle

	\begin{abstract}
		
		In the framework of Algebraic Quantum Field Theory, several operator algebraic notions of entanglement entropy can be associated with any pair of causally disjoint spacetime regions $\mathcal{S}_A$  and $\mathcal{S}_B$ with positive relative distance. Among them, the canonical entanglement entropy is defined as the von Neumann entropy of a canonical intermediate type I factor. In this work, we show that the canonical entanglement entropy of the vacuum state is finite for a broad class of conformal nets including the $U(1)$-current model and the $SU(n)$-loop group models. Since previous studies suggest that this finiteness property is related to nuclearity properties of the system, we show that the mutual information is finite in any local QFT satisfying a modular $p$-nuclearity condition for some $0 < p < 1$. A similar finiteness result is established for another notion of entanglement entropy introduced in this paper. We conclude with remarks for future work in this direction. 
	\end{abstract}
	
	
	
	
	\section{Introduction} \label{sec:1}

    In classical information theory, Shannon entropy provides a quantitative measure of the information content of a given system’s state. Its quantum analogue, the von Neumann entropy, defines the entropy of a normal state of the system via a density matrix and is closely related to quantum entanglement, a phenomenon that occurs in composite systems when a state is not the product of two independent states on the subsystems. In standard non-relativistic quantum mechanics, each normal state $\varphi$ is represented by a density matrix $\rho$, and the von Neumann entropy is given by
	\[
	S(\varphi)=-\text{tr}\, \rho\log \rho \,.
	\]
This notion of entropy is surely consistent whenever the algebra of observables corresponds to a type I factor $\M \cong B(\hil)$, since the density matrix $\rho$ is uniquely determined by $\varphi$. In contrast, within the axiomatic approach to quantum field theory based on Haag–Kastler nets, it is known under very general assumptions that local algebras of observables are type III von Neumann algebras \cite{fredenhagen1985modular}.  In this setting, normal states cannot always be described by density matrices, and the standard von Neumann entropy becomes ill-defined. Therefore, in order to study entanglement phenomena in local QFT, this structural feature necessitates the introduction of alternative entropy-type functionals adapted to type III algebras. The study of entanglement in quantum field theory, already present in early works such as those of Summers and Werner \cite{summers1985vacuum}, has gained renewed attention in recent years due to its deep connections with topics including area theorems \cite{hollands2018entanglement}, c-theorems \cite{casini2007c}, and quantum null energy inequalities \cite{morinelli2021modular, panebianco2021loop}. \\

	In the study of entropy in local QFT, nuclearity conditions play an important role \cite{hollands2018entanglement,narnhofer1994entropy,otani2018toward}. These conditions are based on the compactness criterion proposed by Haag and Swieca \cite{haag1965does}. This criterion is an abstraction of the insight that a QFT model must have a bound on the number of its local degrees of freedom to show regular thermodynamical behaviour. Based on similar heuristic arguments, Buchholz and Wichmann strengthened this assumption and suggested the first nuclearity condition \cite{buchholz1986causal}, which nowadays is known as \emph{energy nuclearity condition}. Explicitly, let $\mathcal{O} \mapsto \mathcal{A}(\mathcal{O}) \subseteq B(\hil)$ be some local Haag-Kastler net describing a local QFT on the Minkowski space. Denote by $\Omega$ the vacuum vector and by $\omega$  the corresponding vacuum state. One says that the energy nuclearity condition holds if  
	\begin{equation} \label{eq:enc0}
		\Theta_{\beta, \mathcal{O}} \colon \mathcal{A}(\mathcal{O}) \to \hil \,, \quad 	\Theta_{\beta, \mathcal{O}} (a) = e^{- \beta H} a \Omega \,,
	\end{equation}
	is nuclear for any bounded region $\mathcal{O}$ and any inverse temperature $\beta >0$. Here $H = P_0$ is the Hamiltonian corresponding to the time direction $x_0$.  The nuclear norm of \eqref{eq:enc0} can be interpreted as the partition function of the restricted system at some fixed inverse temperature, hence it is natural to expect  such nuclear norms to measure  the entropy of the state. Subsequently, a related nuclearity condition has been found by use of  modular theory \cite{buchholz1990nuclear, buchholz1990nuclear2}. According to this second nuclearity condition, one considers an inclusion $\mathcal{O} \subset \widetilde{\mathcal{O}}$ of spacetime regions and requires the map
	\begin{equation}\label{eq:mnc0}
		\Xi \colon \mathcal{A}(\mathcal{O}) \to \hil \,, \quad \Xi(x) = \Delta^{1/4} x \Omega \,,
	\end{equation}
	to be nuclear, with $\Delta$ the modular operator of the bigger local algebra $\mathcal{A}(  \widetilde{\mathcal{O}})$ with respect to $\Omega$. In this case one speaks of {\em modular nuclearity condition}. If the map \eqref{eq:mnc0} is $p$-nuclear then one will say that the {\em modular $p$-nuclearity condition} is satisfied. If modular $p$-nuclearity holds for some $ 0 < p \leq 1$ then the modular nuclearity condition is satisfied, and if so then it is well known that the {\em split property} holds, namely there exists an intermediate type I factor $ \mathcal{A}(  {\mathcal{O}})  \subset \mathcal{R} \subset \mathcal{A}(  \widetilde{\mathcal{O}})$ \cite{buchholz1990nuclear, buchholz1990nuclear2}.  \\
	
	The split property arises as a technical condition in local QFT that ensures the existence of a natural isomorphism between $\mathcal{A}(  {\mathcal{O}}) \vee \mathcal{A}(  \widetilde{\mathcal{O}})'$ and $\mathcal{A}(  {\mathcal{O}}) \otimes \mathcal{A}(  \widetilde{\mathcal{O}})'$, where $\mathcal{O} \subset \widetilde{\mathcal{O}}$ is an inclusion of spacetime regions. For this reason, it has been studied as an expression of the statistical independence of two spacelike separated regions since its first applications \cite{d1983interpolation, doplicher1984standard}. It is considerably stronger than (Einstein) causality, which is included in the Haag-Kastler axioms \cite{hollands2018entanglement}. While causality expresses the independence of measurement apparatuses located in spacelike separated regions, the split property additionally takes the stage of preparation into account. It signifies that no choice of a state prepared in one region can prevent the preparation of any state in the other region. Among all intermediate type I factors, a {\em canonical intermediate type I factor} $\mathcal{F}$ can be chosen by using standard representation techniques \cite{doplicher1984standard}. When the split property holds, one possible entanglement measure to consider is the von Neumann entropy on the canonical intermediate type I factor $\mathcal{F}$, namely what we call the {\em canonical entanglement entropy} \cite{longo2021neumann}. Formally, the canonical entanglement entropy is given by  
	\[
	E_C(\omega)  = S_\mathcal{F} (\omega)\,.
	\]
	In this work we extend some results of \cite{longo2021neumann} and we show this entanglement entropy to be finite on a wide family of conformal nets including the $U(1)$-current. In order to do so, we first investigate where the proof of \cite{longo2021neumann} can be extended and we then construct a vacuum preserving conditional expectation between canonical intermediate type I factors. \\
	
	More generally, it is a common belief that this  entanglement entropy can be bounded from above  by assuming  some nuclearity property of the system. A result of this type is  strongly suggested by  previous works \cite{longo2021neumann, narnhofer1994entropy, otani2018toward} and shows interesting applications in AdS/CFT contexts as  pointed out in \cite{dutta2021canonical}, where it was referred to as {\em splitting entropy} and it was used to conjecture a bound on the reflected entropy. This type of considerations can be extended to a wide family of entanglement measures defined on a generic  bipartite system $A \otimes B$ like the  {\em mutual information} \cite{hollands2018entanglement}
	\[
	E_I(\omega) = S(\omega \Vert \omega_A \otimes \omega_{B}) \,,
	\]
	where $S(\omega \Vert \varphi)$ stands for the Araki's relative entropy between two states $\omega $ and $ \varphi$ \cite{araki1976relative,araki1977relative}. The mutual information quantifies how much knowledge on the second system one can gain from measuring the first one. Here we prove a relation between the mutual information and the modular $p$-nuclearity condition. \\

	The paper is organised as follows. 	In \autoref{sec:2} we collect various notions about modular theory, von Neumann algebras inclusions, quantum entropy, and we describe a few preliminary results such as Theorem \ref{thm:jones}.  In \autoref{sec:3} we  apply the theory of standard representations to construct a cpu map between canonical intermediate type I factors, with a focus on twisted-local nets on the circle. We then use this construction to extend the results of \cite{longo2021neumann} and show the finiteness of the canonical entanglement entropy on a wide family of conformal nets. This section contains the main results: not only the main theorem, namely Theorem~\ref{thm:main}, is of interest in the study of entanglement measures on conformal nets, but the techniques of Theorem~\ref{thm:conditional} can also be of mathematical interest for other issues concerning conformal nets. The mathematical context of \autoref{sec:4} explicitly refers neither to the canonical entanglement entropy nor to conformal nets, but nonetheless it has been motivated by the results of \autoref{sec:3} and related conjectures. Indeed, it is mathematically reasonable to expect the results of Theorem~\ref{thm:main} to still hold on a wider family of conformal nets satisfying some nuclearity property such as the trace-class condition  \cite{longo2021neumann}, and to support this conjecture we prove some similar statement by using techniques of \cite{hollands2018entanglement}. More precisely, we provide explicit upper estimates of the  mutual information and of one more entanglement measure inspired by \cite{otani2018toward} by assuming the modular $p$-nuclearity of the system, with $0 < p < 1$. In \autoref{sec:5} we  add a few little remarks concerning area laws in CFT by following \cite{hollands2018entanglement}  and we apply the  results of \autoref{sec:4}  to a family of  $1+1$-dimensional  models with factorizing S-matrices constructed in \cite{lechner2008construction} in order to investigate the asymptotic behaviour of different entanglement measures as the splitting distance between two causally disjoint wedges diverges. To conclude, in \autoref{sec:6} we present a few remarks that might be useful for future research in this direction.

		\section{Mathematical background} \label{sec:2}
		
			In this section we review some basic notions about modular theory, von Neumann algebras inclusions, relative entropy and other entropy-like state functionals of interest. We also illustrate a few preliminary results achieved on the way. 
	
		\subsection{Von Neumann algebras inclusions} \label{subsec:2.1}

	We begin with a brief overview on relative modular operators. These considerations can be found in the Appendix of \cite{araki1982positive}, or in more recent works like \cite{ceyhan2020recovering} or \cite{ohya2004quantum}. \\
    
    Let $\mathcal{M}$ be a von Neumann algebra in standard form on a Hilbert space $\hi$. We will denote by $\mathcal{P}^\natural$ the natural cone and by $J$ the modular conjugation, and we recall that $J$ is an anti-unitary involution that satisfies $J\zeta=\zeta$ for any $\zeta$ in $\mathcal{P}^\natural$. Let $\varphi$ and $\psi$ be two normal positive linear functionals on $\mathcal{M}$ uniquely represented by the vectors $\xi$ and $ \eta$ of $\mathcal{P}^\natural$. Denote by $s(\varphi)=[\M' \xi]$ and $s(\psi) = [\M' \eta]$ the central supports of $\varphi$ and $\psi$, respectively.   We define the {\em Tomita relative operator}
	\[
	S^0_{\xi, \eta}(x \eta + \zeta) = s(\psi)x^* \xi \,, \quad x \in \M \,, \; \zeta \in [\M \eta]^\perp \,.
	\]
	This conjugate-linear operator is densely defined and closable. Its closure $S_{\xi, \eta}$ admits a polar decomposition
	\[
	S_{\xi, \eta} = J \Delta_{\xi, \eta}^{1/2} \,,
	\]
    where $\Delta_{\xi, \eta}$ is called the {\em relative modular operator}.  In the case $\xi = \eta$, we will write $S_\xi = S_{\xi, \xi}$ and $\Delta_{\xi} = \Delta_{\xi, \xi}$. Furthermore, since vectors in $\mathcal{P}^\natural$ are cyclic if and only if also separating, if $\varphi$ is faithful then $\xi$ is cyclic, and if so then 
	\begin{equation} \label{eq:natural}
	\mathcal{P}^\natural =  \overline{ \big\{ \Delta^{1/4}_\xi  x^*x \xi \,, \, x \in \mathcal{M} \big\} }  \,.
	\end{equation}
	
		Consider now an inclusion $\mathcal{N} \subseteq \mathcal{M}$ of $\sigma$-finite von Neumann algebras and denote by $\omega$ a normal faithful state on $\M$. Our first result is a characterization of any such inclusion admitting an $\omega$-preserving conditional expectation $\varepsilon \colon \M \to \mathcal{N}$. A well known necessary and sufficient condition has been provided by Takesaki (\cite{takesaki2013theory}, Theorem IX.4.2). However, if such conditional expectation exists, then it also satisfies the following structure theorem, which is a result achieved by partially extending  Proposition 3.1.4 and Proposition 3.1.5 of \cite{jones1983index}. 
	
	\begin{thm} \label{thm:jones}
		Let $\mathcal{M}$ be a von Neumann algebra in standard form on  $\hil$. Let $\mathcal{N} $ be a von Neumann subalgebra of $\M$, $\omega$ a normal faithful state of $\M$ and  $\xi$  the unique vector in the natural cone representing $\omega$. If there exists an $\omega$-preserving conditional expectation $\varepsilon \colon \M\to \mathcal{N}$, then the projection $e$ onto $[\mathcal{N} \xi ]$  is the unique  projection in $\mathcal{N}'$ such that

		(i)  $ex \xi = \varepsilon(x)\xi$ and $exe = \varepsilon(x)e$ for $x$ in $\M$,
		
		(ii) $\mathcal{N}e = e \mathcal{M}e$, 
		
		(iii) $\mathcal{N}' = \mathcal{M}' \vee e$.
  
		Furthermore, if we extend $\omega$ by using its representing vector $\xi$ then
		
		(iv)   there is an $\omega$-preserving  $*$-isomorphism $\phi \colon \mathcal{N}e \to \mathcal{N}$ such that $\varepsilon (x)= \phi (exe)$ for $x$ in $\M$. 
	\end{thm}
	\begin{proof}
(i) Since $\xi$ is standard for $\mathcal{M}$ and $\varepsilon$ is an $\omega$-preserving conditional expectation, it can be shown that the map $e(x \xi) = \varepsilon(x) \xi$ naturally extends to an orthogonal projection onto $[\mathcal{N} \xi]$. It is then easy to notice that $e$ belongs to $\mathcal{N}'$, and the identity $exe = \varepsilon(x)e$  follows. The uniqueness of $e$ follows by construction. (ii) This relation is a corollary of the first point.  (iii) Since $e$ belongs to $\mathcal{N}'$, this follows by applying the commutant to both sides of equation $\mathcal{N}e = e \mathcal{M} e$. (iv) We want to define $\phi$ as the inverse of the map $y  \mapsto ye$ for $y$ in $\mathcal{N}$. To prove this map to be an isomorphism, it suffices to show its injectivity. But if $ye=0$ then $ey^*=0$, hence $\omega(y^*y)=0$ which implies $y=0$ by faithfulness. Finally, the state $\omega (\cdot)=(\xi| \cdot \xi)$  satisfies the identity $\omega(exe) = \omega \cdot \phi(exe)$ for $x$ in $\mathcal{N}$. The theorem is proved.
	\end{proof}
	
	\begin{remark}
	If $\omega $ is not faithful, then the previous theorem still holds if $\mathcal{N}$ is assumed to be a von Neumann subalgebra of the reduced algebra $\M_{s(\omega)}=s(\omega) \M s(\omega)$.
	\end{remark}
	
	\begin{cor} \label{cor:natural}
		Let $\M$, $\mathcal{N}$, $\omega $, $\xi$ and $\varepsilon$ be as in the previous theorem. If $\mathcal{P}^\natural_\M $ and $\mathcal{P}^\natural_{\mathcal{N}e} $ are the natural cones of $\M$ and $\mathcal{N}e$ respectively, then
		\[
		e \, \mathcal{P}^\natural_\M  = \mathcal{P}^\natural_{\mathcal{N}e} \subseteq \mathcal{P}^\natural_\M \,.
		\] 
		 In particular, the elements of  $\mathcal{P}^\natural_{\mathcal{N}e} $ correspond to the $\varepsilon$-invariant normal positive functionals of $\M$.
	\end{cor}
\begin{proof}
	We first recall that, by construction,  $\xi$ is a standard vector for $\mathcal{N}e$ on $e \hil$, hence $\mathcal{N}e$ is on standard form on $e \hil$ 
	and it admits a natural cone $\mathcal{P}^\natural_{\mathcal{N}e} \subseteq e \hil$. Let $S = J \Delta^{1/2}$ be the modular operator of $\M$ with respect to $\xi$.  By point (iv) of the previous theorem we have $\varepsilon(x^*) = \varepsilon(x)^*$ and hence $eS = Se$, which implies that $e \Delta = \Delta e$ is the modular operator of $\mathcal{N}e$. We now notice that, since $\omega$ is faithful, its central support is unital and $S$ has dense domain. This implies the identity $eJ=Je$, and by uniqueness of the polar decomposition $Je$ is the modular conjugation of $\mathcal{N}e$. The thesis follows by equation \eqref{eq:natural}, where the last claim is a direct consequence of the uniqueness of the representative vector.
\end{proof}

		\subsection{Quantum entropy} \label{subsec:2.2}
		
		Relative entropy is a useful algebraic tool which provides a non-commutative generalization of the classic Kullback-Leibler divergence and which measures the difference of two states. The first works on the relative entropy are due to Araki \cite{araki1976relative,araki1977relative}. Here we mainly follow \cite{ohya2004quantum} since it also contains other results concerning von Neumann's entropy and conditional entropy. \\

	Let $\varphi$ and $\psi$ be two normal positive functionals on $\M$ as above. The {\em relative entropy} between $\varphi$ and $\psi$ is defined by
	\begin{equation} \label{eq:re}
		S(\varphi \Vert \psi)  =- \int_0^1 \log \lambda \, d(\xi | E_{\eta, \xi}(\lambda) \xi ) - \int_1^\infty  \log \lambda \, d(\xi | E_{\eta, \xi}(\lambda) \xi )  
	\end{equation}
	if  $s(\varphi) \leq s(\psi)$, otherwise $S(\varphi \Vert \psi) = + \infty$ by definition. Here $ \{E_{\eta, \xi}(\lambda) \}_{\lambda \in \R}$ is the spectral family associated to $\Delta_{\eta, \xi}$. Notice that the second integral in \eqref{eq:re} is always finite by the estimate $\log \lambda  \leq \lambda$, hence
	the r.h.s. of \eqref{eq:re} is well defined and by the spectral theorem  $S(\varphi \Vert \psi)$ is finite if and only if $\xi$ belongs to the domain of $|\log \Delta_{\eta, \xi}|^{1/2} $. In particular, if $\xi$ belongs to the domain of $\log \Delta_{\eta, \xi}$ then 
		\begin{equation} \label{eq:re2}
		S(\varphi \Vert \psi)  = -(\xi|\log \Delta_{\eta, \xi} \xi) \,.
	\end{equation}	
Araki's relative entropy is a well defined state functional, namely it does not depend on the choice of the representing vectors. We now recall some properties of the relative entropy \cite{ohya2004quantum}.

	\begin{itemize}
		\item[(r0)] $S(\varphi \Vert \psi)  \geq 0$ if $\varphi$ and $\psi$ are both states, with $S(\varphi \Vert \psi) = 0$ if and only if $\varphi = \psi$.
		
		\item[(r1)] $S(\varphi \Vert \psi) $ is jointly convex. 
        	
        \item[(r2)] $S(\varphi \Vert \psi) \leq S(\varphi \Vert \omega)$ if $\psi \geq \omega$, and $S(\varphi \Vert \lambda \psi) = S(\varphi \Vert \psi) - \log \lambda$ for $\lambda>0$.
		
		\item[(r3)] $S(\varphi \Vert \psi)$ is monotone increasing with respect to inclusions of von Neumann algebras. 
		
		\item[(r4)] Let  $(\M_i)_{i }$ be  an increasing net of von Neumann subalgebras of $\M$ with the property  $(\cup_i \M_i)'' = \M$. Then the increasing net $ S_{\M_i}(\varphi \Vert \psi)  $ converges to $S(\varphi \Vert \psi) $. 
		
		\item[(r5)] Let $\varphi$ be a normal state on the spatial tensor product  $\M_1 \otimes \M_2$ with partials $\varphi_i = \varphi \vert_{\M_i}$. Given  normal states $\psi_i$  on $\M_i$, we have $S(\varphi \Vert \psi_1 \otimes \psi_2) = S(\varphi_1 \Vert \psi_1 ) + S(\varphi_2  \Vert \psi_2 ) + S(\varphi \Vert \varphi_1 \otimes \varphi_2)$. 
	\end{itemize}

	The relative entropy can be extended to any  $C^*$-algebra $A$ by using its universal representation $\pi$. Indeed, since every state $\omega$ of $A$ admits a normal extension $\tilde{\omega} $ to $A^{**}\cong \pi(A)''$, then for states $\varphi$ and $\psi$ of $A$ one simply defines $S(\varphi \Vert \psi) = S(\tilde{\varphi} \Vert \tilde{\psi})$.  This generalization leads us to the following definition.

	\begin{defn}
		If $\varphi$ is a state on a $C^*$-algebra $A$, then the {\em von Neumann entropy} of $\varphi$ is defined by
		\begin{equation} \label{eq:vne}
			S_A(\varphi) = \sup\Big\{  \sum_i \lambda_i S(\varphi_i \Vert \varphi) \colon \sum_i \lambda_i \varphi_i  = \varphi \Big\} \,,
		\end{equation}
		where the supremum is over all decompositions of $\varphi$ into finite (or equivalently countable) convex combinations of other states. If $A$ is implicit from the context, then we will simply write $S_A(\varphi) = S(\varphi)$. 
	\end{defn}
	Some properties of $S(\varphi)$ are immediate from those of the relative entropy: it is non-negative and it vanishes if and only if $\varphi$ is a pure state. The von Neumann entropy is a concave and strongly subadditive state-functional as described in \cite{ohya2004quantum}. On type I factors, the von Neumann entropy of a normal state $\varphi$ with density matrix $\rho$ is given by  $S(\varphi) = - \text{tr} \, \rho \log \rho$. As a corollary, a simple computation shows that for a normal state $\varphi_1 \otimes \varphi_2$ on $B(\hil_1) \otimes B(\hil_2)$ we have 
	\begin{equation} \label{eq:tensor}
		S(\varphi_1 \otimes \varphi_2) = S(\varphi_1) + S(\varphi_2) \,.
	\end{equation}

	\begin{defn}
		Consider an inclusion of $C^*$-algebras $ A \subseteq B$ and a state $\varphi$ on $B$. The {\em entropy of $\varphi$  with respect to $A$}, also called the {\em subalgebra entropy of $A$ with respect to $\varphi$}, is defined as 
		\begin{equation} \label{eq:ce}
			H_\varphi^B(A) = \sup \Big\{ \sum_i \lambda_i S_A(\varphi_i  \Vert \varphi )  \colon \varphi = \sum_i \lambda_i \varphi_i \Big\} \,,
		\end{equation}
		where the supremum is over all finite (countable) convex decompositions $\varphi = \sum_i \lambda_i \varphi_i$ on $B$. If the bigger $C^*$-algebra $B$ is clear from the context, we will briefly use the notation $H_\varphi^B(A) = H_\varphi(A) $.
	\end{defn}
	The subalgebra entropy \eqref{eq:ce} generalizes the von Neumann entropy \eqref{eq:vne} and is  a particular case of what is known as conditional entropy \cite{connes1987dynamical, ohya2004quantum}. We list a few of its basic properties. 
	\begin{itemize}
		\item[(s0)] $H_\varphi^{{B_2} }(A_1) \leq H_\varphi^{B_1}({A_2} )$ if $A_1 \subseteq {A_2} \subseteq B_1 \subseteq {B_2} $.
		
		\item[(s1)]  $ \lim_i H_\varphi^B(A_i) = H_\varphi^B(A) $ if $(A_i)_i$ is an increasing net of $C^*$-subalgebras of $B$ with union norm dense in $A$.
		
		\item[(s2)] $\lambda H^B_{\varphi_1} (A) + (1 - \lambda) H^B_{\varphi_2} (A) \leq H^B_{\varphi} (A) \leq  \lambda H^B_{\varphi_1} (A) + (1 - \lambda) H^B_{\varphi_2} (A)  + \eta(\lambda) + \eta(1-\lambda)  $ for $\varphi = \lambda \varphi_1 + (1- \lambda) \varphi_2$ on $B$ and $\lambda$ in $(0,1)$, where $\eta(t)=-t \log t$.
	\end{itemize}

	Statement (s1) still holds if all the $C^*$-algebras are replaced with von Neumann algebras whenever the union $\cup_i A_i$ is strongly dense in $A$  and  the state $\varphi$ is normal. We note that the concavity of $H_\varphi(A)$ mentioned in (s2) certainly holds whenever $A$ is AF (\cite{ohya2004quantum}, Theorem 5.29 and Proposition 10.6), but the general case is a bit unclear to the authors \cite{connes1987dynamical}. What is clear  instead,	is the  following   original simple lemma which says whenever the  inequality (s0) reduces to  an equality in the case $A_1 = {A_2}$.

	\begin{lem} \label{lem:inequality}
		Consider the $C^*$-algebras inclusions $A_1 \subseteq B_1 \subseteq {B_2}$ and $A_1 \subseteq {A_2} \subseteq {B_2}$. Let  $\varphi$ be a state on ${B_2}$. If there is a $\varphi$-preserving conditional expectation $\varepsilon \colon {B_2}  \to {B_1} $, then 
		\[
		H_\varphi^{B_1}(A_1) \leq H_\varphi^{{B_2}}({A_2}) \,.
		\]
	\end{lem}
	
	\begin{proof}
		We  follow Proposition 6.7 of \cite{ohya2004quantum}. If  $\psi$ is a state of $B_1$ then $\psi \cdot \varepsilon $ is a state of ${B_2}$. Therefore, if  $\varphi = \sum_i \lambda_i \varphi_i $ on $B_1$ for some states  $\varphi_i$ of  $B_1$ then $\varphi = \sum_i \lambda_i \varphi_i \cdot \varepsilon $  is a decomposition of $\varphi$ into states of ${B_2}$. The rest follows from $S_{A_1}(\varphi_i \Vert \varphi) \leq S_{{A_2}}(\varphi_i \cdot \varepsilon  \Vert \varphi) $ which is a consequence of (r3).
	\end{proof}

\begin{cor} \label{cor:inequality}
	Consider an inclusion of $C^*$-algebras $A \subseteq B$ and a state $\varphi$ on $B$. If there is a $\varphi$-preserving conditional expectation $\varepsilon \colon {B}  \to {A} $, then $S_A(\varphi) = H_\varphi^{{B} }(A) $ and $S_A(\varphi) \leq S_B(\varphi)$.
\end{cor}
\begin{proof}
    The first identity is straightforward from property (s0) and the previous lemma, and the second inequality follows.
\end{proof}

\section{Entanglement entropy in chiral CFT} \label{sec:3}

We begin with a few definitions. Let $\mathcal{K}$ be the family of all the open, nonempty and non dense intervals of the circle.  For $I$ in $\mathcal{K}$, $I'$ denotes the interior of the complement. The M\"obius group $\text{M\"ob} $ acts on the circle by linear fractional transformations. A {\em M\"obius covariant net} $(\mathcal{A}, U, \Omega)$ consists of a family $\{\mathcal{A}(I) \}_{I \in \mathcal{K}}$ of von Neumann algebras acting on a complex Hilbert space $\hil$, a strongly continuous unitary representation $U$ of M\"ob and a vector $\Omega$ in $\hil$, called the {\em vacuum vector}, satisfying the following properties  \cite{d2001conformal}:
\newline 

(a0) $\al(I_1) \subseteq \al(I_2)$ if $I_1 \subseteq I_2$ (isotony),

(a1) $U(g) \al(I) U(g)^* = \al(g.I)$ for every $g $ in M\"ob and $I$ in $\mathcal{K}$ (M\"obius covariance),

(a2)  the representation $U$ has {\em positive energy}, namely the generator of rotations has non-negative spectrum (positivity of the energy),

(a3) $\Omega$ is cyclic for the von Neumann algebra $ \bigvee_{I \in \mathcal{K}} \al(I)$, and up to a scalar $\Omega$ is the  unique M\"ob-invariant vector of $\hil$ (vacuum). \\

A M\"obius covariant net is said to be {\em twisted-local} if the following axiom is satisfied:  \\

(a4) there exists a self-adjoint unitary operator $Z$ commuting with the representation $U$, such that $Z\Omega = \Omega $ and  $Z \al(I')Z^* \subseteq \al(I)'$   (twisted-locality).\\

In the following, twisted-local M\"obius covariant nets will be briefly referred to as {\em twisted local nets}. A twisted-local net is said to be {\em local} if $Z$ is the identity. Two consequences of the axioms (a0)-(a3) are  \cite{d2001conformal} \\

(a5) $\Omega$ is cyclic and separating for each $\al(I) $  (Reeh-Schlieder property),

(a6) $\al(I) \subseteq \bigvee_{\alpha} \al(I_\alpha)$ if $I \subseteq \bigcup_\alpha I_\alpha$ (additivity), \\

while if we also assume (a4) then we have   \cite{d2001conformal} \\

(a7) $\al(I') = Z\al(I)'Z^*$ for every $I$ in $\mathcal{K}$ (twisted-duality),

(a8) if $I_+$ is the upper half of the circle and $\Delta$ is the modular operator associated to $\al(I_+)$ and $\Omega$, then for every $t$ in $\R$ we have 
\begin{equation} \label{eq:BW}
	\Delta^{it} = U (\delta_{-2 \pi t}) \,,
\end{equation}
where $\delta$ is the one-parameter group of dilations (Bisognano-Wichmann). \\

Finally, the last axiom we will need is the following:\\

(a9) if ${I} $ and $ \tilde{I}$ are two open intervals verifying $\bar{I} \subset \tilde{I}$, then there exists a type I factor $\mathcal{R}$ such that $\al(I) \subset \mathcal{R} \subset \al(\tilde{I})$ (split property). \\

By axiom (a5) every intermediate type I factor $\mathcal{R}$ is separable, hence the split property ensures the separability of the underlying Hilbert space $\hil$.  In the  local nets setting, it can be proved that the split property is automatic in the case of {\em conformal nets}   \cite{morinelli2018conformal}. As described in the next section, between all intermediate type I factors a  {\em canonical intermediate type I factor} can be uniquely defined. More precisely, if $J$ is the modular conjugation of the von Neumann algebra $\mathcal{A}(I) \vee \mathcal{A}(\tilde{I})'$, which is in standard form by the Reeh-Schlieder property,  then the canonical intermediate type I factor can be proven to be the unique $J$-invariant intermediate type I factor \cite{doplicher1984standard}. Explicitly, such a type I factor is given by 
\begin{equation} \label{eq:canononical-intermediate}
    \mathcal{F} = \mathcal{A}(I) \vee J \mathcal{A}(I)J = \mathcal{A}(\tilde{I}) \cap J \mathcal{A}(\tilde{I}) J \,.
\end{equation}

 \begin{defn}
     Let $(\mathcal{A}, U, \Omega)$ be a twisted-local net on a Hilbert space $\hil$. We call a family $\mathcal{B} = \{  \mathcal{B}(I) \}_{I \in \mathcal{K}}$ of von Neumann subalgebras $\mathcal{B}(I) \subseteq \mathcal{A}(I) $ a {\em subnet of $\mathcal{A}$} if  it satisfies isotony and M\"obius covariance with respect to $U$.  We will use the notation $\mathcal{B} \subseteq \mathcal{A}$ to denote a subnet  $\mathcal{B} $ of $ \mathcal{A}$.
 \end{defn}

  An important feature about subnets is the existence of a family of conditional expectations between local algebras. More precisely,  for each interval $I$ there is a canonical vacuum-preserving faithful normal conditional expectation $\varepsilon_I \colon \mathcal{A}(I) \to \mathcal{B}(I)$ by Bisognano-Wichmann for $\mathcal{A}$, M\"obius covariance for $\mathcal{B}$ and Takesaki's theorem (\cite{takesaki2013theory}, Theorem IX.4.2.). Notice that Theorem \ref{thm:jones} applies to any of these von Neumann algebras inclusions.

\begin{prop} \label{prop:conditional}
	Let $\mathcal{B} \subseteq \mathcal{A}$ be a subnet of a twisted-local net $\mathcal{A}$. All the normal conditional expectations $\varepsilon_I \colon \mathcal{A}(I) \to \mathcal{B}(I)$ extend to a unique vacuum-preserving conditional expectation $\varepsilon \colon \mathfrak{A} \to \mathfrak{B}$, where   the $C^*$-algebra $\mathfrak{A} $ is the  norm closure of the union of all the local algebras  $\mathcal{A}(I)  $ and similarly for $ \mathfrak{B}$. Then, the projection $e=[\mathfrak{B} \Omega]$ satisfies $e(x \Omega) = \varepsilon(x) \Omega$ and $exe = \varepsilon(x)e$ for all $x$ in $\mathfrak{A}$. 
\end{prop}
\begin{proof}
Since $\mathfrak{B}' = \bigcap_{I \in \mathcal{K}}\mathcal{B}(I)'$,  $e$ belongs to the commutant of all the von Neumann algebras $\mathcal{B}(I)$, and by Theorem \ref{thm:jones} we have that $e(x \Omega) = \varepsilon_I(x) \Omega$ for  $x$ in $\mathcal{A}(I) $ and for each interval $I$ in $\mathcal{K}$. As the vacuum state  $\omega(\cdot) = (\Omega | \cdot \Omega)$ is locally faithful by Reeh-Schlieder, the isotony property for $\mathcal{B}$ implies that the conditional expectations are compatible, namely $\varepsilon_{\tilde{I}}$ extends $\varepsilon_I$ whenever $I \subset \tilde{I}$. Therefore, we can define $\varepsilon(x)  $  for $x$ in  $\bigcup_{I \in \mathcal{K}} \mathcal{A}(I)$ by setting $\varepsilon(x) = \varepsilon_I(x)  $ whenever $x$ is in  $\mathcal{A}(I)$. The map $\varepsilon$ is bounded since every $\varepsilon_I  $ has unital norm, hence  we can continuously extend $\varepsilon$ to $\mathfrak{A}$ (this procedure is also known as the BLT theorem). Finally, $\varepsilon$ is a conditional expectation since by  continuity $\varepsilon$ is a positive $ \mathfrak{B}$-linear projection, and the identity $\omega = \omega \cdot \varepsilon$ follows as well. To prove the last statement, one first shows that $x \Omega \mapsto \varepsilon(x) \Omega$, with $x$ in $\mathfrak{A}$, is a well defined projection onto $e \hil$. Finally, the identity $exe = \varepsilon(x)e$ follows \cite{jones1983index}. 
\end{proof}

We now notice that the Jones projection $e$ onto $\hil_{\mathcal{B}} = \overline{\mathfrak{B}\Omega}$ of the previous proposition commutes not only with all the local algebras  $\mathcal{B}(I)$, but with the unitary representation $U$ as well. Therefore,  we can state that   $\mathcal{B}$ can be naturally identified with a M\"obius covariant net  on $\hil_\mathcal{B} $ together with the  restriction of $U$ on $\hil_{\mathcal{B}} $ and the vacuum vector $\Omega$. Indeed, axioms (a1) and (a2) are verified since $\hil_\mathcal{B}$ is $U$ invariant and axiom (a3) follows by Proposition \ref{prop:conditional}. As a corollary, as noticed in Theorem \ref{thm:jones} we have that  $\mathcal{B}(I)$ is  isomorphic to  $\mathcal{B}(I)e = e\mathcal{A}(I)e$. 

\begin{lem} \label{lem:twist}
	Let $\mathcal{B} \subseteq \mathcal{A}$ be an inclusion  of  M\"obius covariant  nets with Jones projection $e=[\mathfrak{B} \Omega]$. Assume $(\mathcal{A}, U, \Omega)$ to be twisted-local with twist operator $Z_{\mathcal{A}}$. If $eZ_{\mathcal{A}}=Z_{\mathcal{A}}e$, then  $\mathcal{B}$  is twisted-local. 
\end{lem}
\begin{proof}
	We show $Z_{\mathcal{B}} = eZ_{\mathcal{A}}e$ to be a twist operator for $\mathcal{B} $. Clearly $Z_{\mathcal{B}}$ fixes the vacuum vector and commutes with $U$. Then, by Theorem \ref{thm:jones} and twisted locality we have  the chain of  inclusions $Z_{\mathcal{B}} \mathcal{B}(I')Z_{\mathcal{B}}^* =  Z_{\mathcal{B}} \mathcal{A}(I')Z_{\mathcal{B}}^* \subseteq e \mathcal{A}(I)' e \subseteq e \mathcal{B}(I)' e $. The thesis follows.
\end{proof}


\begin{defn} \label{def:ce}
	\cite{longo2021neumann} Let $(\mathcal{A}, U, \Omega)$ be a twisted-local net satisfying the split property. Given a proper inclusion  of open intervals $I \subset \tilde{I}$ of the circle such that $\bar{I} \subset \tilde{I}$, denote by $\mathcal{F} $ the corresponding canonical intermediate type I factor. We define the {\em canonical entanglement entropy} of $\omega$ with respect to the inclusion $I \subset \tilde{I}$ as the von Neumann entropy 
	\begin{equation} \label{eq:cee}
			E_C(\omega) = S_\mathcal{F}(\omega) =  S_{\mathcal{F}'}(\omega) \,.
	\end{equation}
\end{defn}

\begin{thm} \label{thm:conditional}
	Let $(\mathcal{A}, U, \Omega)$ be a twisted-local net on some Hilbert space $\hil$ with twist operator $Z_{\mathcal{A}}$. Consider a subnet  $\mathcal{B} \subseteq \mathcal{A}$ and assume the induced Jones projection $e$ to commute with $Z_\mathcal{A}$. Assume also $\mathcal{A}$ and $\mathcal{B}$ to both  satisfy the split property. Denote then  by  $\mathcal{F}_{\mathcal{A}} $ and $ \mathcal{F}_{\mathcal{B}}$   the   canonical intermediate type I factors associated to the inclusions $\al(I) \subset   \al(\tilde{I})$  and $\mathcal{B}(I) \subset   \mathcal{B}(\tilde{I})$  respectively, with $\bar{I} \subset \tilde{I}$ as above.  If $\varepsilon$ and $e$ are as in Proposition \ref{prop:conditional}, then $\varepsilon(\mathcal{F}_{\mathcal{A}})e = \mathcal{F}_{\mathcal{B}} $. In particular, 
	\begin{equation} \label{eq:inequality}
		S_{\mathcal{F}_{\mathcal{B}}}(\omega) \leq S_{\mathcal{F}_{\mathcal{A}}}(\omega) \,.
	\end{equation}
\end{thm}
\begin{proof}
		We denote by $J_{\mathcal{A}} $ the modular conjugation of $\mathcal{A}(I) \vee \mathcal{A}(\tilde{I})'$ and by $J_{\mathcal{B}}$ the modular conjugation of $\mathcal{B}(I) \vee \mathcal{B}(\tilde{I})'$. By Proposition  \ref{prop:conditional} we have a conditional expectation $\varepsilon \colon \mathcal{A}(\tilde{I}) \to  \mathcal{B}(\tilde{I}) $, but by Lemma \ref{lem:twist} we also have that $e\mathcal{A}({I})'e = Z_\mathcal{A} \mathcal{B}({I}') Z_\mathcal{A}^*e $. Therefore, by Theorem \ref{thm:jones} the map $\varepsilon$ restricts to a conditional expectation $\varepsilon \colon \mathcal{A}({I})' \cap \mathcal{A}(\tilde{I}) \to  Z_\mathcal{A} \mathcal{B}({I}') Z_\mathcal{A}^* \cap \mathcal{B}(\tilde{I}) $, and thanks to the  considerations described in the proof of Corollary \ref{cor:natural}  we can claim that $eJ_{\mathcal{A}}=J_{\mathcal{A}}e = J_{\mathcal{B}}e$. We can now state that the conditional expectation $\varepsilon$ maps $  \mathcal{A}(\tilde{I}) \cap J_\mathcal{A} \mathcal{A}(\tilde{I})J_\mathcal{A}  $ onto $  \mathcal{B}(\tilde{I}) \cap J_{\mathcal{A}} \mathcal{B}(\tilde{I})J_\mathcal{A}  $:  clearly $\varepsilon(\mathcal{A}(\tilde{I})) = \mathcal{B}(\tilde{I})$, but thanks to Theorem \ref{thm:jones} and the  previous remark we also have  $\varepsilon(J_{\mathcal{A}} \mathcal{A}(\tilde{I})J_{\mathcal{A}}) = J_\mathcal{A} \mathcal{B}(\tilde{I})J_\mathcal{A} $. By identity \eqref{eq:canononical-intermediate}, it follows that  $\varepsilon(\mathcal{F}_{\mathcal{A}})e = \mathcal{F}_{\mathcal{B}} $. Finally, inequality \eqref{eq:inequality} is a consequence of assertion (iv) of Theorem \ref{thm:jones} and Corollary \ref{cor:inequality}.
\end{proof}

	Theorem \ref{thm:conditional} is a crucial  technical result of this section. In addition, even though we focused on twisted-local nets on the circle, Proposition \ref{prop:conditional} and Theorem \ref{thm:conditional} easily extend to any inclusion  of double cones $\mathcal{O} \subset \tilde{O}$ in any inclusion $\mathcal{B}\subseteq \mathcal{A}$  of conformal field theories on $\R^{d+1}$, since in CFT the modular group of the double cones local algebras has a geometric implementation and hence a vacuum preserving conditional expectation \cite{hislop1982modular}.

\begin{thm} \label{thm:main}
	Let $(\mathcal{A}, U, \Omega)$  be one of the following twisted-local  nets on the circle:
	
	(i) the free Fermi net,

	(ii) the $U(1)$-current, 
	
	(iii) any $LSU(n)$-conformal net of level $\ell \geq 1$, 
	
	(iv) any Virasoro net with  central charge  given by 
	\begin{equation} \label{eq:central}
		c = \frac{\ell (n^2 -1)}{\ell + n}  \,, 
	\end{equation}
with $\ell \geq 1$ and $n \geq 2$ integers. 

 If $ \mathcal{F} $  is  the  canonical intermediate type I factor given by the inclusion  $\mathcal{A}(I) \subseteq   \mathcal{A}(\tilde{I})$   with $\overline{I} \subseteq \tilde{I}$, then 
	\begin{equation}  \label{eq:finite}
		S_{\mathcal{F}}(\omega) < + \infty \,.
	\end{equation}
\end{thm}
\begin{proof}
	(i) The finiteness property \eqref{eq:finite} has been proved on the free Fermi net  in \cite{longo2021neumann}, where explicit estimates can be found. (ii) In this case we exploit the equivalent definition of the $U(1)$-current conformal net as a subnet of the free Fermi net \cite{longo2021neumann}. Indeed, even though in  the free Fermi net graded locality rather than locality holds, we can still apply Theorem \ref{thm:conditional} because in this case the Jones projection is the identity. (iii) By similarity of the first quantization Hilbert spaces, the same proof provided in \cite{longo2021neumann} for the free Fermi net  can be replicated on  any $LSU(n)$-conformal net of level $\ell =1$  \cite{panebianco2021loop, wassermann1998operator}. The explicit estimates provided in \cite{longo2021neumann} still apply in this setting.   Therefore, since the embedding $LSU(n) \subseteq LSU(n \ell) $ gives rise to all the  $LSU(n)$-conformal nets of level $\ell \geq 1$  \cite{wassermann1998operator}, the estimates of \cite{longo2021neumann} apply to any $LSU(n)$-conformal net. (iv)	Since  $LSU(n)$-conformal nets are local, we can apply Theorem \ref{thm:conditional} to any conformal subnet of these models like  the Virasoro net with central charge given by \eqref{eq:central} \cite{panebianco2021loop}. 
\end{proof}

Theorem \ref{thm:main}, which can be summarized as  a generalization of \cite{longo2021neumann}  by using Theorem \ref{thm:conditional},  is the main result of this work. It is worth noting that the proof exhibited in  \cite{longo2021neumann} heavily depends on the structure of the free Fermi net. However, more generally the finiteness property \eqref{eq:finite} is expected to rely on some nuclearity condition of the system such as the trace-class property \cite{hollands2018entanglement, otani2018toward}. To support this conjecture,  in the next section we provide a few results in this direction.

\section{Modular nuclearity and entanglement measures} \label{sec:4}


In this section we review a few general notions about modular nuclearity conditions of interest for applications to local QFT. We then provide  explicit estimates of the mutual information by assuming a certain modular nuclearity condition to be satisfied. For an axiomatic treatment of entanglement measures and a deeper study of their connection with nuclearity conditions, we refer to \cite{hollands2018entanglement}.

	\subsection{Modular nuclearity conditions} \label{subsec:4.1}

Consider a couple $A$, $B$  of commuting von Neumann algebras on some Hilbert space $\hil$. We shall say that the pair $(A, B)$ is {\em split} if there exists a von Neumann algebra isomorphism $\phi \colon A \vee B \to A \otimes B$ such that $\phi(ab) = a \otimes b$. We will refer to split pairs also as {\em bipartite systems}.  Clearly the spatial tensor product $A \otimes B$ has a natural structure of bipartite system since $ A \cong A \otimes \mathbbm{1}$ and  $B \cong \mathbbm{1} \otimes B$.  If $ A \vee B$ is $\sigma$-finite, then the pair $(A, B) $ is split if and only if for any given normal states $\varphi_A$ on $A$ and $\varphi_B$ on $B$ there exists a normal state $\varphi$ on $A \vee B$ such that  $\varphi (ab) = \varphi_A(a) \varphi_B(b)$ \cite{longo2020lectures}. In local quantum field theory, usually bipartite systems  appear from local algebras associated to causally disjoint and distant spacetime regions \cite{hollands2018entanglement}.  \\

We will say that a split pair $(A, B)$ is {\em standard} if $A$, $B$ and $A \vee B$ are in standard form with respect to some vector $\Omega$. Standard split pairs will also be denoted by $(A,B,\Omega)$. We  set $J_{ A} = J_{A, \Omega}$, $J_B = J_{B, \Omega}$, and similarly $\Delta_A = \Delta_{ A , \Omega}$, $\Delta_B = \Delta_{ B , \Omega}$. As $A \otimes B$ is in standard form with respect to  $\Omega \otimes \Omega$,  the isomorphism $\phi \colon A \vee B \to A \otimes B$ has a {\em standard implementation}, namely is uniquely implemented by some unitary $U$ which maps the natural cone of $A \vee B$ onto the natural cone of $A \otimes B$  \cite{doplicher1984standard}. It can also be shown that $J_{ A} \otimes J_{B} = U J U^{-1} $, with $J = J_{ A \vee B, \Omega}$. The {\em canonical intermediate type I factors} are $F= U^{-1}(B(\hil) \otimes \mathbbm{1})U$ and $F' = U^{-1}( \mathbbm{1} \otimes B(\hil))U$. It can be shown that $F $ is the unique $J$-invariant type I factor $A \subseteq {F} \subseteq B'$, and similarly for $F'$. If $A$ and $B$ are both factors then ${F} = A \vee J A J = B' \cap J B' J$, and therefore ${F'} = B \vee J B J = A' \cap J A' J$ \cite{doplicher1984standard}.  \\

 An inclusion   $N \subseteq M$  of von Neumann algebras  is said to be {\em split} if the pair $(N, M')$ is split.  We shall often pass from a split inclusion to a
 split pair and back. The trivial inclusion $N=M$ is split if and only if $N$ is a type I factor 	\cite{longo2020lectures}. The inclusion  $N \subseteq M$ is said to be {\em standard} if there is a vector $\Omega$ which is standard for $N$, $M$ and the relative commutant $N' \cap M$. If $N \subseteq M$ is a standard split inclusion then each intermediate type I factor $R$ is $\sigma$-finite and hence separable, therefore the Hilbert space $\hil$ has to be separable as $R \Omega$ is dense in $\hil$. If  $N \vee M'$ has a cyclic and separating vector, then the pair $(N, M')$ is split if and only if  there is an intermediate type I factor  $N  \subseteq R \subseteq M$ \cite{longo2020lectures}.

\begin{defn}
	Consider an inclusion $N \subseteq M$ of von Neumann algebras on a Hilbert space $\hil$. Assume the existence of a standard vector $\Omega$ for $M$ and denote by $\Delta$ the corresponding modular operator. We will say that the inclusion $N \subseteq M$  satisfies the {\em modular nuclearity condition} if the map
	\begin{equation} \label{eq:xi}
		\Xi \colon N \to \hil \,, \quad \Xi(x) = \Delta^{1/4} x \Omega \,,
	\end{equation}
	is nuclear. 
\end{defn}

A modular nuclear  inclusion of factors is split, and a split  inclusion of factors implies the map \eqref{eq:xi} to be compact  \cite{buchholz1990nuclear}. This motivates the interest in the split property in local quantum field theory contexts, where the split property amounts to some form of statistical independence between causally disjoint spacetime regions \cite{hollands2018entanglement, lechner2008construction}. \\

The previous nuclearity condition can be easily generalized as follows. Consider a linear map $\Theta \colon \mathcal{E} \to \mathcal{F}$ between Banach spaces. The map $\Theta$ is said to be {\em $p$-nuclear}, where $p>0$, if there are a  sequence of linear functionals $e_i \in \mathcal{E}^*$ and a sequence of elements $f_i$ in $\mathcal{F}$ such that 
\begin{equation}\label{eq:pnc}
			\Theta(x) = \sum_i e_i(x)f_i \,, \quad x \in \mathcal{E}, \quad  \sum_i \Vert e_i \Vert^p \Vert f_i \Vert^p < + \infty \,. 
\end{equation}
The induced quasi-norm, also called {\em $p$-norm}, is given by
\begin{equation*}
	 \Vert \Theta \Vert_p = \inf  \Big( \sum_i \Vert e_i \Vert^p \Vert f_i \Vert^p \Big)^{1/p}  \,,
\end{equation*}
where the infimum is over all  possible representations of $\Theta$ of the form \eqref{eq:pnc}. A $p$-nuclear map is also $q$-nuclear for all $0<p \leq q \leq 1$, and $1$-nuclear maps are nuclear by the very definition. Finally, the above nuclearity condition  can be  rephrased as a {\em modular $p$-nuclearity condition} if the map \eqref{eq:xi} is $p$-nuclear. 

\begin{defn}
	Let $(A,B,\Omega)$ be a standard split pair. Denote by $\Delta_A$ and $\Delta_B$ the corresponding modular operators. We define 
	\begin{equation} \label{eq:sn}
	\Xi_A(a) = \Delta_{B'}^{1/4} a \Omega  	\,, \quad \Xi_B(b) = \Delta_{A'}^{1/4} b \Omega   \,.
	\end{equation}
	with $a$ in $A$ and $b$ in $B$. Given $p>0$ we define the {\em $p$-partition function} as 	
	\begin{equation} \label{eq:nu}
		z_p = \min \{ \Vert \Xi_A \Vert_p , \Vert \Xi_B \Vert_p \}  \,.
	\end{equation}
	We will say that $(A,B,\Omega)$ satisfies the {\em $p$-modular nuclearity condition} if the $p$-partition function is finite. In the case $p=1$ we will simply talk of {\em partition function} and  of {\em modular nuclearity condition}.
\end{defn}

The  $p$-modular nuclearity condition implies the modular nuclearity condition if $p \leq 1$. In order to motivate our definition, we notice that if $(A, B)$ satisfies the  modular nuclearity condition then it is split.

	\begin{defn}
		A local quantum field theory on the Minkowski space is said to satisfy the {\em modular nuclearity condition (for double cones)} if the inclusion  $\mathcal{A}(\mathcal{O}_1) \subseteq \mathcal{A}(\mathcal{O}_2) $ is  modular nuclear  whenever ${\mathcal{O}_1} \subset \mathcal{O}_2$ is an inclusion of double cones such that $\overline{\mathcal{O}_1} \subset \mathcal{O}_2$. If the inclusions $\mathcal{A}(\mathcal{O}_1) \subseteq \mathcal{A}(\mathcal{O}_2) $ are modular $p$-nuclear, then we will say that the {\em modular $p$-nuclearity condition (for double cones)} is satisfied.
	\end{defn}

		This notion of modular nuclearity condition in local quantum field theory is what motivates the previous abstract definitions of modular nuclearity conditions. Nuclearity conditions are mathematical properties ensuring the model to exhibit a regular thermodynamic behaviour. Conformal nets satisfying the trace-class property and the free scalar field are examples of QFT models satisfying the modular $p$-nuclearity condition for any $0<p \leq 1$ \cite{buchholz2007nuclearity, lechner2016modular}. We will further discuss this topic in  \autoref{sec:5}.

	\subsection{Entanglement upper estimates} \label{subsec:4.2}

 A state $\omega$ on a bipartite system $A \otimes B$ is said to be {\em separable} if there are positive normal functionals $\varphi_j$ of $A$ and $\psi_j$ of $B$ such that $\omega = \sum_j \varphi_j \otimes \psi_j$, where the sum is assumed to be norm convergent.  A normal state which is not separable is called {\em entangled}. Separable states are normal and satisfy the following original lemma inspired by \cite{narnhofer2002entanglement}. 
 
 \begin{lem} \label{lem:entropy}
 	Given two von Neumann algebras $A$ and $B$, consider a state  $\omega$  on the bipartite system $A \otimes B$. If $\omega$ is separable, then
 	\[
 	S_A(\omega) = H_\omega^{A \otimes B}(A) \,.
 	\]
 \end{lem}
 
 \begin{proof}  If $\omega = \sum_j \phi_j \otimes \psi_j$ and $\pi$ is the GNS representation of $A$ associated to the marginal state $\omega_A = \omega \vert_A$, then we have an equivalence of GNS representations $\pi \cong  \oplus_j \pi_j$, where $\pi_j$  is the GNS representation of $A$ given by $\psi_j(1) \phi_j$. We can define a cpu map $\varepsilon_j \colon A \otimes B \to \pi_j (A)$ by $\varepsilon_j(a \otimes b) = \pi_j(a) \psi_j(b)/\psi_j(1)$, and this leads us to define a conditional expectation $\varepsilon \colon A \otimes B \to pAp$ by $\varepsilon = \pi^{-1} \cdot \oplus_j \varepsilon_j$, where $p$ is the support projection of $\omega_A$. By Corollary \ref{cor:inequality} we have $S_{pAp}(\omega) = H_\omega^{A \otimes B}(pAp)$. Before concluding this lemma, we notice that $S_A(\varphi \Vert \omega) = S_{p A p}(\varphi \Vert \omega)$ whenever $\varphi$ is a normal state such that $\varphi(p)=1$, which is certainly true for all states $\varphi_i$ appearing in a convex decomposition $\omega = \sum\lambda_i \varphi_i$ (see \cite{ohya2004quantum} for details). This implies $S_A(\omega) = S_{pAp}(\omega) $, and  the thesis follows.
 \end{proof}

\begin{defn} \label{defn:mi}
	The {\em mutual information} $E_I(\omega)$ of a state $\omega$ on $A \otimes B$ is given by 
	\begin{equation} \label{eq:mi0}
			E_I(\omega) = S(\omega \Vert \omega_A \otimes \omega_B) \,.
	\end{equation}
	where $\omega_A= \omega \vert_A$ and similarly for $B$. 
\end{defn}


As motivated by relative entropy's property (r0), the mutual information is a measure of how the state $\omega$ differs from the separable state $\omega_A \otimes \omega_B$. The mutual information is non-negative and independent of the order of $A$ and $B$ \cite{hollands2018entanglement}. For separable states $\omega = \sum_j \lambda_j  \varphi_j \otimes \psi_j$, with $\varphi_j$ and $\psi_j$ normal states, we have
\begin{equation} \label{eq:e1}
	E_I(\omega) \leq \sum_j \eta(\lambda_j) \,,
\end{equation}
with $\eta(t) = -t \ln t$ the information function. Indeed, properties (r0), (r1), (r2) and (r5) imply that
\[
S(\omega \Vert \omega_A \otimes \omega_B) \leq \sum_j \lambda_j S(\varphi_j \otimes \psi_j\Vert \varphi_j \otimes \omega_B) = \sum_j \lambda_j S( \psi_j\Vert \omega_B) \leq \sum_j \lambda_j S( \psi_j\Vert \lambda_j \psi_j) = \sum_j \eta(\lambda_j) \,.
\]
In order to further characterize the mutual information on finite-dimensional bipartite systems, let us consider  a bipartite system given by  $A=B(\hil)$ and $B=B(\hil')$, with $\hil$ and $\hil'$ finite-dimensional Hilbert spaces. In this finite-dimensional setting, the mutual information is finite and given by
\begin{equation} \label{eq:mi}
	E_I(\omega) = S(\omega_A) + S(\omega_B) - S (\omega) \,.
\end{equation}
Furthermore, if $\omega$ is pure then the partials $\omega_A$ and $\omega_B$ have density matrices with equal nonzero eigenvalues, and hence $S(\omega_A) = S(\omega_B)$ (see Lemma 6.4. of \cite{ohya2004quantum} for details). Therefore,
\begin{equation} \label{eq:pure}
	E_I(\omega) = 2S(\omega_A) = 2S(\omega_B) \,.
\end{equation}
We point out that, without any finiteness assumption, on hyperfinite type I factors we can only write $	E_I(\omega) + S (\omega)  = S(\omega_A) + S(\omega_B) $. Other properties of the mutual information can be found in \cite{hollands2018entanglement}. Moreover, we can use \eqref{eq:mi} to deduce the following lemma.

\begin{lem}\label{lem:concave-mutual-info}
	Let $A$ and $B$ be hyperfinite factors and $\omega=\sum_j \lambda_j \omega_j$ be a convex decomposition of a normal state $\omega$ on $A\otimes B$, with $\omega_j $ normal states. Then 
	\begin{align*}
		 E_I (\omega) \leq \sum_j \lambda_j E_I(\omega_j) +2 \sum_j \eta(\lambda_j) \,, \qquad \text{and} \qquad 	\sum_j \lambda_j E_I(\omega_j) \leq E_I (\omega)  + \sum_j \eta(\lambda_j)\,.
	\end{align*}
\end{lem}
\begin{proof}
	This result is true for finite-dimensional type $I$ factors thanks to \eqref{eq:mi} and conditional entropy's property (s2). If $A$ and $B$ are generic hyperfinite factors, then there is a family of bipartite systems of finite-dimensional factors whose union is weakly dense in $A\otimes B$. The statement follows by the approximation property (r4) of the relative entropy.
\end{proof}

In the previous section we proved the canonical entanglement entropy to be finite on some twisted-local nets on the circle. Works on this topic suggest such an entanglement measure  to be  finite by assuming some modular nuclearity condition of the model  such as the trace-class condition \cite{hollands2018entanglement, longo2021neumann}. Even though a general proof on a model independent ground is still lacking, in this section we provide a few  results in this direction. The mathematical strategy, inspired by \cite{hollands2018entanglement}, is very natural and consists in writing the vacuum state $\omega$ as a difference of two separable functionals thanks to some modular nuclearity of the system, and then to use concavity properties typical of entanglement measures to provide explicit estimates. We begin with an easy extension of a lemma contained in  \cite{hollands2018entanglement}.

	\begin{lem}\label{lem:HS-lem3}
		Let $(A, B, \Omega)$ be a standard split pair, with $\Omega$ inducing a state $\omega$. We  assume modular $p$-nuclearity to hold for some $0 < p \leq 1$, namely the $p$-partition function $z_p$ defined in \eqref{eq:nu} is finite for some $0 < p \leq 1$. Given $\epsilon >0$, there are sequences of normal linear functionals $\phi_j$ on $A$ and $\psi_j$ on $B$ such that 
		\begin{equation} \label{eq:sep}
			\omega(ab) = \sum_j \phi_j(a) \psi_j(b) \,, \quad a \in A \,, \; \;  b \in B \,,
		\end{equation}
		and $\sum_j \Vert \phi_j \Vert^p \Vert \psi_j \Vert^p < z_p^p + \epsilon $.
	\end{lem}
	
	\begin{proof}
	 We assume $z_p = \Vert \Xi_A \Vert_p$ and we follow Lemma 3 of \cite{hollands2018entanglement}. Let us consider $\Delta = \Delta_{B', \Omega}$ and $J=J_{B, \Omega}$. Since $J\Delta^{-1/2}$ is the Tomita operator for $B$, for $b$ in $B$ we have
\[
\Delta^{-1/2}b^*\Omega=J(J\Delta^{-1/2})b^*\Omega = Jb\Omega = JbJ\Omega \,.
\]
  Therefore, given $a$ in $A$, we can show that
  \begin{equation*}
			\begin{split}
				\omega(ab) & = (\Omega | ab \Omega) = ((\Delta^{1/4} + \Delta^{-1/4})^{-1}(1+\Delta^{-1/2})b^* \Omega | \Delta^{1/4} a \Omega) \\
				& = ((\Delta^{1/4} + \Delta^{-1/4})^{-1}(b^* + JbJ)\Omega | \Xi_A(a)) \,.
			\end{split}
		\end{equation*}
		If $z_p$ is finite and $\epsilon >0$, then there are sequences of positive normal functionals $\phi_j$ on $A$ and vectors $\xi_j$ in $\hil$ such that
		\[
		\Xi_B(a) = \sum_j \phi_j(a) \xi_j \,, \quad a \in A \,,
		\]
		and $\sum_j \Vert \phi_j \Vert^p \Vert \xi_j \Vert^p < z_p^p + \epsilon$. Define now normal functionals $\psi_j$ on $B$ by
		\begin{center}
			$  			\psi_j(b) = ((\Delta^{1/4} + \Delta^{-1/4})^{-1}(b^* + JbJ)\Omega| \xi_j) \,,  $
		\end{center}
		and note that $\Vert \psi_j \Vert \leq \Vert \xi_j \Vert$ because of the estimate $\Vert (\Delta^{1/4} + \Delta^{-1/4})^{-1} \Vert \leq 1/2$ and the spectral calculus. Putting both paragraphs together we find the conclusion.
	\end{proof}

	Before providing a corollary of the previous lemma, we describe a general procedure known as {\em polarization} of a functional. Let $\omega$ be a continuous functional.  We will say that $\omega$ is {\em self-adjoint} if $\omega = \omega^*$, with $\omega^*(x) = \overline{\omega(x^*)}$ the {\em conjugate} of $\omega$.  By use of   $\omega^*$ one can write $\omega = \phi + i \psi$, with $\phi$ and $\psi$ self-adjoint. Then, after applying a Jordan decomposition on both $\phi$ and $\psi$, we can write $\omega= \sum_{\alpha=0}^3 (i)^\alpha \omega_\alpha$, with $\omega_\alpha$ positive. The inequality $\Vert \omega_\alpha \Vert \leq \Vert \omega \Vert$ can also be proved, and $\omega_\alpha$ are all normal if $\omega$ is.

\begin{cor} \label{cor:four}
	With the notation of the previous lemma, for every $\epsilon > 0$ we can write $\omega = (1+\lambda) \omega_+ - \lambda \omega_-$, where $\omega_\pm$ are separable states and $(1+\lambda)^p \leq 4 (z_p^p + \epsilon)$. 
\end{cor}

\begin{proof}
	By polarization, we can decompose $\phi_j = \sum_{\alpha=0}^3 (i)^\alpha \phi_j^\alpha$ and  $\psi_j = \sum_{\alpha=0}^3 (i)^\alpha  \psi_j^\alpha$ in four positive normal functionals. As mentioned above,  $ \Vert \phi_j^\alpha \Vert \leq \Vert \phi_j \Vert $ holds for any $\alpha$, and similarly for $\psi_j$. Since $\omega$ is positive, then after the identification $A \vee B \cong A \otimes B$ we find 
	\[
	\omega = \sum_j \sum_{\alpha=0}^{3} \phi_j^\alpha \otimes \psi_j^{4 - \alpha}- \sum_j \sum_{\alpha=0}^{3} \phi_j^\alpha  \otimes \psi_j^{2 - \alpha} \,,
	\]
	namely $\omega$ is difference of two separable functionals. The thesis follows.
\end{proof}

	\begin{lem}\label{lem:HS-lem4}
		With the hypotheses of  Lemma \ref{lem:HS-lem3}, assume $\omega$  to have an expression like in \eqref{eq:sep} and assume $\mu_p = \sum_j \Vert \phi_j \Vert^p \Vert \psi_j \Vert^p $ to be finite for some $0 < p \leq 1$. Then there is a separable positive linear functional $\sigma$ such that $\sigma \geq \omega$ on $A \vee B$ and $\Vert \sigma \Vert^p = \mu_1^p \leq \mu_p $. 
	\end{lem}
	
	\begin{proof}
		We follow \cite[Lemma 4]{hollands2018entanglement}. By polar decomposition there are partial isometries $u_j$ in $A$ such that $\phi(u_j \cdot ) \geq 0$ on $A$  and $\phi_j(u_j u_j^* \cdot ) = \phi_j$. It follows in particular that $\phi_j (u_j) = \Vert \phi_j \Vert$ and 
		\begin{center}
			$  		\bar{\phi}_j(a) = \overline{\phi_j(u_j u_j^* a^*)} = \phi_j (u_j (u_j^*a^*)^*) = \phi_j(u_j a u_j)  $
		\end{center}
		for all $a$ in $A$, where we used the fact that $\phi_j(u_j \cdot )$ is hermitian (here $\bar{\psi}(a) = \overline{\psi(a^*)}$). Similarly, there are partial isometries   $v_j$ in $B$ such that $\psi_j (v_j \cdot ) \geq 0$ and $\psi_j (v_j v_j^* \cdot) = \psi_j$. Note that the positive linear functional $\rho_j = \phi_j (u_j \cdot ) \otimes \psi_j (v_j \cdot)$ is separable. Writing $w_j = u_j \otimes v_j$ we then define
		\[
		\sigma_j (\cdot) = \frac{1}{2}\rho_j(\cdot) + \frac{1}{2}\rho_j( w^* \cdot w) \,,
		\]
		which is also separable, because $w$ is a simple tensor product. Furthermore, 
		\begin{center}
			$  	\Vert \sigma_j \Vert = \sigma_j(1) = \rho_j(1) = \Vert \phi_j \Vert  \Vert \psi_j \Vert\,,   $
		\end{center}
		and also
		\[
		0 \leq \frac{1}{2} \rho_j((1-w^*) \cdot (1-w)) = \sigma_j - \frac{1}{2} (\phi_j \otimes \psi_j + \bar{\phi}_j \otimes \bar{\psi}_j ) \,.
		\]
		We conclude that  $\sigma = \sum_j \sigma_j$ is a separable positive linear functional with 
		\[
		\sigma \geq \frac{1}{2} \sum_j  (\phi_j \otimes \psi_j + \bar{\phi}_j \otimes \bar{\psi}_j ) = \frac{1}{2} ({\omega} + \overline{\omega}) = \omega \,. 
		\]
		and $ \Vert \sigma \Vert^p =  \big( \sum_j \Vert \sigma_j \Vert \big)^p \leq  \sum_j \Vert \sigma_j \Vert ^p =  \mu_p $. 
	\end{proof}
	
	Once again, the previous lemma is a slight extension of a lemma proved in \cite{hollands2018entanglement}. 		Notice that, by the two previous lemmas, we have $z_p \geq 1$. We now take advantage of the previous estimates to provide the main theorem of this section.

	\begin{thm} \label{thm:mutual}
		Let $(A, B,\Omega)$ be a standard split pair of hyperfinite factors, with $\Omega$ representing a state $\omega$. Assume the $p$-partition function $z_p$ defined in \eqref{eq:nu} to be finite for some $0 < p <  1$. Then 
		\begin{equation} \label{eq:mt}
			E_I(\omega) \leq c_p z_p + \eta(z_p-1) - \eta(z_p)  \,,
		\end{equation}
		where  $c_p = \frac{1}{(1-p)e}$ and $\eta(t) = -t \ln t$.
	\end{thm}

	\begin{proof}
		We begin the proof by recalling that, as claimed in Lemma \ref{lem:concave-mutual-info}, if   $\omega = \sum_j \lambda_j \omega_j$ is  a finite convex decomposition in states then
		\begin{equation} \label{eq:concave}
		E_I (\omega) \geq \sum_j \lambda_j E_I(\omega_j) - \sum_j \eta(\lambda_j) \,.
		\end{equation}
	Moreover, by  the previous lemmas  for every $\epsilon >0$ we have a separable functional $\sigma \geq \omega$ such that $\Vert \sigma \Vert ^p \leq z_p^p + \epsilon $. By setting $\hat{\sigma} = \sigma / \Vert \sigma \Vert $  we can write $\omega = \Vert \sigma\Vert \hat{\sigma} - \Vert \tau\Vert \hat{\tau}$ and apply \eqref{eq:concave} to notice that
		\[
		\Vert \sigma \Vert E_I(\hat{\sigma}) \geq E_I(\omega) + \eta (\Vert \sigma\Vert) - \eta (\Vert \sigma\Vert - 1) \,,
		\] 
		where we used the positivity  property $E_I(\hat{\tau}) \geq 0$. We then recall that the separability of $\hat{\sigma}$ implies that $E_I(\hat{\sigma}) \leq  \Vert \sigma \Vert ^{-p} c_p (z_p^p + \epsilon) $. Therefore, the claimed estimate follows from the inequality $\eta(t) \leq c_p t^p$ for $p<1$  and  the monotonicity of $\eta(s-1) - \eta(s)$ for $s \geq 1$. 
	\end{proof}

	\begin{remark}
Due to the  inequality $S(\varphi \Vert \omega) \geq \Vert \varphi - \omega \Vert^2/2$ \cite{ohya2004quantum}, we can use the previous result to estimate the distance between $\omega$ and $\omega \otimes \omega$. See \cite{morinelli2018conformal} for related issues regarding the split property.
	\end{remark}

	This upper bound of the mutual information, which by monotonicity becomes sharper for smaller $p$,  has been inspired by the works  \cite{hollands2018entanglement, narnhofer2002entanglement}. Lower bounds on the mutual information follow from Theorem 17 and  Theorem 18 of \cite{hollands2018entanglement}, while an explicit computation has been provided in CFT contexts in \cite{longo2018relative}. We recall that modular $p$-nuclearity for double cones holds on scalar free fields for any $p>0$ \cite{lechner2016modular} and on a wide family of $1+1$-dimensional  models with factorizing S-matrices \cite{alazzawi2017inverse, lechner2008construction, lechner2018approximation}. On conformal nets on the circle satisfying the trace-class property, we also have that modular $p$-nuclearity holds for open intervals inclusions  \cite{buchholz2007nuclearity}. In addition, all these models  satisfy the hyperfiniteness of the local algebras, hence Theorem \ref{thm:mutual} can be applied in these settings. We now follow \cite{otani2018toward} and we show the finiteness of some “tailored” entanglement entropy under the assumption of modular $p$-nuclearity for some $0  < p < 1$. 
	
	\begin{defn} \label{def:ip} Let $(A, B, \Omega)$ be a standard split pair of von Neumann algebras on a Hilbert space $\hil$. If  $u \colon \hil \to \hil \otimes \hil$ is a unitary implementing the natural  isomorphism $A \vee B \cong A \otimes B$, then we denote by $R_u = u^{-1}(B(\hil) \otimes \mathbbm{1} )u$ the induced  type I factor. We call  an {\em intermediate pair} any such pair $(u, R_u)$.
	\end{defn}

	\begin{defn} 
		Let $(A, B,\Omega)$ be a standard split pair as in the previous definition. Given a state $\psi$ on $B(\hil)$, we  call the {\em intermediate entanglement entropy} of $\psi$ the functional
		\begin{equation} \label{eq:iee}
			I(\psi)= \sup_{(u, R_u)} \inf_{\phi, \lambda} \frac{1}{\lambda} S(\phi)\,.
		\end{equation}
		Here the supremum is over all intermediate pairs, the  infimum is over all states $\phi$ on $B(\hil)$ and real numbers $0 < \lambda \leq 1$ such that $\phi \geq \lambda \psi$ on $A \vee B$,  and $S(\phi)$ is the von Neumann entropy of $\phi$ on  $R_u$.
	\end{defn}

	\begin{thm} \label{thm:otani} 	Let $(A, B,\Omega)$ be a standard split pair, with $\Omega$ inducing a state $\omega$. Denote by $z_p$ the $p$-partition function \eqref{eq:nu}.  If $z_p$ is finite for some $0 < p < 1$, then  
		\begin{equation}\label{eq:otani}
			I(\omega)\leq  z_p \ln z_p +  c_p z_p^p  \, .
		\end{equation}
	\end{thm} 
	
	\begin{proof}
		The proof consists of a computation that does not depend on the choice of the intermediate pair, which is therefore implicit in what follows. Through the natural isomorphism $A \vee B \cong A \otimes B$ we will  identify $A$ with $A \otimes \mathbbm{1}$ and $B$ with $\mathbbm{1}\otimes B $. Lemma \ref{lem:HS-lem4} gives a separable dominating normal functional $\sigma  \geq \omega$ with $\Vert \sigma \Vert^p \leq z_p^p + \epsilon$ for $\epsilon >0$ arbitrarily small.  We utilize the separability of $\hat{\sigma} = \sigma / \Vert \sigma \Vert $ over the bipartite system $  A \otimes B $ and decompose it into positive, normal functionals, say  $\hat{\sigma}=\sum_j \phi_j\otimes\psi_j$. Without loss of generality we can assume $\phi_j$ to be states on $A$. Now we notice that $\phi_j \otimes \psi_j $ is a normal positive functional on $A \otimes B$, hence it can be extended by taking a representative vectors. Since such extension has same norm, we can extend $\sigma$ to a separable positive functional on $B(\hil) \otimes B(\hil)$ in such a way that still $\Vert \sigma \Vert^p \leq z_p^p + \epsilon$.  If we now set  $\eta(t) = - t \ln t $ and  $1/c_p = {(1-p)e}$, then
		\begin{equation*} \label{eq:est1}
			\Vert \sigma  \Vert S(\hat{\sigma})  \leq \Vert \sigma  \Vert \ln \Vert \sigma  \Vert + \sum_j \eta (\Vert \psi_j \Vert) \leq \Vert \sigma  \Vert \ln \Vert \sigma  \Vert + c_p \sum_j \Vert \psi_j \Vert^p \,.
		\end{equation*}
		The claimed estimate follows as $\sigma$ is arbitrary. 
	\end{proof}

	\section{Asymptotic in local QFT}  \label{sec:5}

		In this  section we provide a few simple considerations about the asymptotic behaviour of the previously studied entanglement measures, namely the canonical entanglement entropy \eqref{eq:cee}, the mutual information \eqref{eq:mi0}, and the intermediate entanglement entropy \eqref{eq:iee}. \\

Thanks to the results of \cite{hollands2018entanglement}, we can partially answer to a problem proposed in \cite{longo2021neumann}, namely we  show  the canonical entanglement entropy to satisfy a lower bound of area law type.  Let $(\mathcal{A}, U, \Omega)$ be any of the nets on the circle of Theorem \ref{thm:main} except the free Fermi net, so that locality holds. Consider two disjoint and distant open intervals $I$ and $J$ of $S^1 \setminus \{ -1\}$. By stereographic projection we can identify $I$ and $J$ with two open intervals $\tilde{I}$ and $\tilde{J}$ of the real line which are assumed to be bounded \cite{panebianco2021loop}. Set $s = \text{dist}(\tilde{I}, \tilde{J})$ and denote by $E_C(s)$ the corresponding canonical entanglement entropy. If we denote by $E_I(s)$ the mutual information of the vacuum on $\mathcal{A}(I) \vee \mathcal{A}(J)$, then the monotonicity property of the relative entropy implies that $E_I(s) \leq 2 E_C(s)$. Therefore, by Theorem 17 of \cite{hollands2018entanglement} we can claim that 
			\begin{equation} \label{eq:al}
				E_C(s)  \apprge  \frac{D_2}{2}  \cdot \ln \frac{\min(|\tilde{I}|, |\tilde{J}|)}{s} \,,
			\end{equation}
where $D_2$ is a constant explicitly defined in \cite{hollands2018entanglement}. The meaning of the notation above is the following: if $f(s)$ and $g(s)$ are two non-negative real functions depending on a real variable $s > 0$, we will say that $f(s) \apprle g(s)$ as $s \to 0$ if for any $\delta >0 $ there is some $s_0>0$ such that $f(s) \leq (1+\delta)g(s)$ for all $s$ in $(0,s_0)$. In higher dimension, we can consider the 1+1 dimensional chiral CFTs induced by the previous conformal nets as described in \cite{hollands2018entanglement}. On these models the finiteness of the canonical entanglement entropy still holds thanks to equation \eqref{eq:tensor}. We can then exploit Theorem 17 of \cite{hollands2018entanglement} once again  to provide other asymptotic lower estimates which better motivates the connection between the canonical entanglement entropy and area laws typically appearing in such settings. \\

We proceed with a couple of remarks about the asymptotic behaviour of the mutual information and the intermediate entanglement entropy on a bipartite system appearing from two causally disjoint and distant spacetime regions $\mathcal{S}_A$ and $\mathcal{S}_B$ in a local QFT. Usually, some boundedness assumption on at least one region between $\mathcal{S}_A$ and $\mathcal{S}_B$ is a mandatory assumption, since  the split property does not hold for unbounded regions  in more than two spacetime dimensions \cite{buchholz1974product, lechner2008construction}. In two spacetime dimensions instead, a class of integrable QFT models that satisfy modular nuclearity for wedges was constructed in \cite{lechner2008construction}. The only input needed are the mass $m>0$ and a $2$-body scattering matrix $S_2$, even though such a nuclearity property has been proved only if the scattering matrix is of fermionic type and satisfies some regularity assumptions, and if the splitting distance $s>0$ is large enough  \cite{lechner2008construction}. However, if this is the case then  further works prove that  modular $p$-nuclearity  for any $p>0$  for wedges holds as well  \cite{alazzawi2017inverse, lechner2018approximation}. Furthermore, if $\Xi(s)$ is the associated modular operator, then it can be shown that $\Vert \Xi(s) \Vert_p \to 1$ as $s \to + \infty$. We summarize these results in the following theorem.

	\begin{thm} \label{thm:gandalf} \cite{alazzawi2017inverse, lechner2008construction, lechner2018approximation} Let $(\mathcal{A}, U, \Omega)$ be an integrable quantum field theory on $\R^2$ with fermionic and regular factorizing $S$-matrix $S_2$. Define
		\begin{equation} \label{eq:nuclear}
			\Xi (s) \colon \mathcal{A}(W_R)  \to \hil \,, \quad \Xi(s)A = \Delta^{1/4} U({\bf s}) A \Omega \,, \quad s>0 \,,
		\end{equation}
		where $W_R$ is the right wedge  $x_1 > |x_0|$, $U({\bf s})$ is the unitary associated to the translation of ${\bf s} = (0,s)$ and $\Delta$ is the modular operator of $(\mathcal{A}(W_R), \Omega)$. Then, there exist some finite  splitting distance $s_{\text{min}} < \infty$ such that $\Xi(s) $ is $p$-nuclear for all $p>0$ and $s > s_{\text{min}} $. In addition,  $\Vert \Xi(s) \Vert_p \to 1$ as $s \to + \infty$.
	\end{thm}

This theorem tells us that the embedding $\mathcal{A}(W_R + {\bf s}) \subseteq \mathcal{A}(W_R)$ satisfies modular $p$-nuclearity for all $p>0$ if the splitting distance $s>0$ is sufficiently large.  Therefore, if we denote by $E_I(s)$ the mutual information associated to the left wedge $W_L=W_R'$ and to $W_R + {\bf s}$, by  applying the upper bound \eqref{eq:mt} for any $p>0$  we can state that
		\[
		\limsup_{s \to + \infty} E_I(s) \leq 1/e \,.
		\]
		By monotonicity, the same asymptotic upper estimate holds for any couple of double cones contained in $W_L$ and $W_R + {\bf s}$ respectively. Analogously, we can apply \eqref{eq:otani} to estimate the intermediate entanglement entropy in the limit $s \to + \infty$. With the same arguments and similar notation, we find
				\[
		\limsup_{s \to + \infty} I(s) \leq 1/e \,.
		\]

	
	

	\section{Outlook} \label{sec:6}
	
	We close this paper with a miscellanea of additional remarks that might be useful for future research in this area. In particular, we list a few conditions which are equivalent to the finiteness of the canonical entanglement entropy. \\

	The techniques of  \autoref{sec:5} rely on the presence of a  separable state $\sigma$ on the bipartite system $A\otimes B $ that dominates $\omega$. If $F$ is an intermediate type I factor $A\subseteq F \subseteq B'$ arising from the natural isomorphism $A \vee B \cong A \otimes B $ as in Definition \ref{def:ip}, then it  is possible to construct a separable functional on $  B(\hil) \otimes B(\hil) \cong F \vee F' $  that dominates $\omega$ on $F $ and on $F'$ by use of generalized conditional expectations \cite{accardi1982conditional}. More specifically, let $(A,B,\Omega)$ be a standard split pair with finite $p$-partition function for some $0<p < 1$. As discussed in  \cite{accardi1982conditional, ohya2004quantum}, one has two $\omega$-preserving cpu maps, say $\varepsilon$  and $\varepsilon'$, induced by the inclusions  $A\subseteq F$ and $B \subseteq  F'$ respectively.  By the isomorphism $  B(\hil) = F \vee F' \cong B(\hil) \otimes B(\hil)  $ we can then define a cpu map $\gamma$ on $ B(\hil) $ such that $\gamma(xy)=\varepsilon(x)\varepsilon'(y)$ for $x$ in $F$ and $y$ in $F'$. If $\sigma$ is the dominating separable functional from Lemma \ref{lem:HS-lem4}, then ${\sigma}_0=\sigma \cdot \gamma$ dominates $ {\omega}_0= \omega \cdot  \gamma$. Notice that $\omega = {\omega_0} $ on $F$ and on $F'$, but generally not on $B(\hil)$. The functional ${\sigma_0} = \sum_j \phi_j\cdot\varepsilon \otimes  \psi_j\cdot\varepsilon'$   is separable with $\sum_j \lVert\phi_j\cdot\varepsilon\rVert^p \lVert\psi_j\cdot\varepsilon'\rVert^p=\mu_p  $ finite  (cf. Lemma \ref{lem:HS-lem4} for notation), and by proceeding as in Theorem~\ref{thm:mutual} we have 
	\[
	E_I({\omega}_0) = S({\omega_0}\Vert \omega_F \otimes \omega_{F'}) \leq c_p z_p + \eta(z_p-1) - \eta(z_p)  \,,
	\]
	where the r.h.s. is finite by assumption. Unfortunately, this does not imply the finiteness of the canonical entanglement entropy since ${\omega_0}$ is not a pure state on $B(\hil)$. But we can make use of generalized conditional expectations to give an equivalent description of the canonical entanglement entropy. In particular, by use of equation \eqref{eq:pure} and Lemma \ref{lem:entropy} we can claim  that 
	\[
	 S_{B(\hil)}(\omega \Vert \omega_F \otimes \omega_{F'})= 2 E_C(\omega)   = 2 H_{{\omega_0}} (F)  \,,
	\]
	with $H_{{\omega}_0} (F) = H_{{\omega}_0}^{B(\hil)} (F) $ the conditional entropy. The authors of \cite{dutta2021canonical} argued on grounds of physical arguments that
	\[
	E_C(\omega) \approx E_I^{F \vee B'}(\omega) = S(\omega \Vert \omega_F \otimes \omega_{B'}) \,,
	\] 
	and indeed it is reasonable to expect that the results of this work can be properly strengthened. For example, Theorem \ref{thm:mutual} implies that $E_I^{F \vee B'}(\omega)$ is finite if the $\omega$-preserving generalized conditional expectation from $F \vee  B' $ onto $ A\vee B'$ is a separable operation in the terminology of  \cite{hollands2018entanglement}. Another  strategy could be that of estimating the entanglement entropy of some energy cutoff of the vacuum state like  in \cite{otani2018toward} and then to operate some limit procedure.  A different approach is the one of \cite{longo2021neumann}, in which the authors use the language of standard subspaces. Unfortunately, even if completely rigorous, this last work heavily depends on the structure of the free Fermi nets, and a generalization of it seems quite challenging up to now. However, in the context of conformal nets the authors expect the trace-class property to be a good assumption to start with \cite{longo2020lectures, otani2018toward}.  \\
	
	{\bf Acknowledgements } 	We thank Roberto Longo for suggesting us the problem. We also thank Yoh Tanimoto and Gandalf Lechner for explanations concerning the last section and Yoh Tanimoto, Gandalf Lechner and Ko Sanders for comments on an earlier version. \\
	
		{\bf Fundings } 	BW is part of the INdAM Doctoral programme in Mathematics and/or Applications cofunded by  Marie  Sklodowska-Curie actions (INdAM-DP-COFUND 2015) and has received funding from the European Union's 2020 research and innovation programme under the Marie Sklodowska-Curie grant agreement No 713485. \\
		
					{\bf  Data availability} Data sharing is not applicable to this article as no datasets were generated or analyzed during the current study. \\
		
			{\bf Conflict of interest } On behalf of all authors, the corresponding author states that there is no conflict of interest.

\end{document}